\documentclass[11pt]{article}
\usepackage{amsthm}
\usepackage{fullpage}

\usepackage{comment,amsfonts,amssymb,amsmath,algorithmic,algorithm, microtype}

\usepackage{graphicx}
\usepackage{wrapfig}
\usepackage{xspace}

\newcommand{\namedref}[2]{\hyperref[#2]{#1~\ref*{#2}}}

\newcommand{\theoremref}[1]{\namedref{Theorem}{#1}}

\newcommand{\figureref}[1]{\namedref{Figure}{#1}}

\newcommand{\eqnref}[1]{\namedref{Equation}{#1}}

\newcommand{\R}{\mathbb{R}}

\newcommand{\zo}{\{0,1\}}
\providecommand{\card}[1]{\lvert#1\rvert}
\providecommand{\abs}[1]{\left\lvert#1\right\rvert}

\providecommand{\aset}[1]{\{#1\}}
\providecommand{\tuple}[1]{\langle{#1}\rangle}
\providecommand{\eqdef}{:=}
\providecommand{\pred}[1]{\ensuremath{\operatorname{\textsf{#1}}}}
\providecommand{\poly}{\mathrm{poly}}
\DeclareMathOperator{\Val}{Val}
\DeclareMathOperator{\sk}{sk}
\DeclareMathOperator{\est}{est}
\DeclareMathOperator{\indic}{\mathbf{1}}

\def\compactify{\itemsep=0pt \topsep=0pt \partopsep=0pt \parsep=0pt}
\newcommand{\Benczur}{Bencz{\'{u}}r\xspace}

\usepackage{ifpdf}
\ifpdf    
\usepackage[colorlinks,linkcolor=blue,filecolor=blue,citecolor=blue,urlcolor
=blue,pdfstartview=FitH]{hyperref}
\else    
\usepackage[colorlinks,linkcolor=black,filecolor=black,citecolor=black,urlcolor=black,hypertex]{hyperref}
\fi

\newcommand{\alert}[1]{\textbf{\color{red}
    [[[#1]]]}\marginpar{\textbf{\color{red}**}}\typeout{ALERT:
    \the\inputlineno: #1}}

\newtheorem{theorem}{Theorem}
\numberwithin{theorem}{section}

\newtheorem{remark}[theorem]{Remark}

\newtheorem{definition}[theorem]{Definition}
\newtheorem{observation}[theorem]{Observation}

\floatstyle{ruled}
\newfloat{algorithm}{tbp}{loa}
\providecommand{\algorithmname}{Algorithm}
\floatname{algorithm}{\protect\algorithmname}

\makeatother

\usepackage{pdfsync}

\begin{document}
\title{Sparsification of Two-Variable Valued CSPs}

\author{Arnold Filtser\thanks{%
    Ben-Gurion University of the Negev, Israel.
    Partially supported by the Lynn and William Frankel Center for Computer Sciences.
    Email: \texttt{arnoldf@cs.bgu.ac.il}
  }
  \and
  Robert Krauthgamer\thanks{%
    Weizmann Institute of Science, Israel. 
    Work supported in part by the Israel Science Foundation grant \#897/13 and the US-Israel BSF grant \#2010418.
    Email: \texttt{robert.krauthgamer@weizmann.ac.il}}
}

\maketitle

\begin{abstract}  
	
A valued constraint satisfaction problem (VCSP) instance $(V,\Pi,w)$ 
is a set of variables $V$ with a set of constraints $\Pi$ weighted by $w$.
Given a VCSP instance, we are interested in a re-weighted sub-instance 
$(V,\Pi'\subset \Pi,w')$ such that preserves the value of the given instance
(under every assignment to the variables) within factor $1\pm\epsilon$.
A well-studied special case is cut sparsification in graphs, 
which has found various applications.
We show that a VCSP instance consisting of a single boolean predicate $P(x,y)$ (e.g., for cut, $P=\pred{XOR}$) can be sparsified into $O(|V|/\epsilon^2)$ constraints if and only if the number of inputs that satisfy $P$ is anything but one (i.e., $|P^{-1}(1)| \neq 1$). 
Furthermore, this sparsity bound is tight unless $P$ is a relatively trivial predicate.
We conclude that also systems of 2SAT (or 2LIN) constraints can be sparsified.
\end{abstract}

\section{Introduction}
\label{sec:intro}

The seminal work of \Benczur and Karger \cite{BK96} showed that 
every edge-weighted undirected graph $G=(V,E,w)$ admits 
cut-sparsification within factor $(1+\epsilon)$ 
using $O(\epsilon^{-2} n\log n)$ edges, where we denote throughout $n=\card{V}$.
To state it more precisely, assume that edge-weights are always non-negative
and let $\pred{Cut}_{G}(S)$ denote the total weight of edges in $G$ 
that have exactly one endpoint in $S$.
Then for every such $G$ and $\epsilon\in (0,1)$,
there is a re-weighted subgraph 
$G_{\epsilon}=(V,E_{\epsilon}\subseteq E,w_{\epsilon})$
with $|E_\epsilon|\le O(\epsilon^{-2} n\log n)$ edges, such that
\begin{equation}  \label{eq:cut_sparsifier}
  \forall S\subset V, \qquad 
  \pred{Cut}_{G_{\epsilon}}(S)
  \in(1\pm\epsilon)\cdot\pred{Cut}_{G}(S) ,
\end{equation}
and moreover, such $G_\epsilon$ can be computed efficiently.

This sparsification methodology turned out to be very influential.
The original motivation was to speed up algorithms for cut problems
-- one can compute a cut sparsifier of the input graph 
and then solve an optimization problem on the sparsifier
-- and indeed this has been a tremendously effective approach,
see e.g.\ \cite{BK96,BK02,KL02,sherman2009breaking,madry2010fast}. 
Another application of this remarkable notion is to reduce space requirement, 
either when storing the graph or in streaming algorithms \cite{AG09}.
In fact, followup work offered several refinements, improvements, 
and extensions 
(such as to spectral sparsification or to cuts in hypergraphs,
which in turn have more applications) 
see e.g.~\cite{ST04a,ST11,SS11,dCHS11,FHHP11,KP12,NR13,BSS14,KK15}.
The current bound for cut sparsification is $O(n/\epsilon^2)$ edges,
proved by Batson, Spielman and Srivastava~\cite{BSS14}, 
and it is known to be tight \cite{ACKQWZ15}.

We study the analogous problem of sparsifying Constraint 
Satisfaction Problems (abbreviated CSPs), 
which was raised in \cite[Section 4]{KK15} and goes as follows.
Given a set of constraints on $n$ variables, 
the goal is to construct a sparse sub-instance,
that has approximately the same value as the original instance
under \emph{every possible assignment},
see Section~\ref{sec:predicates} for a formal definition.
Such sparsification of CSPs can be used to reduce storage space 
and running time of many algorithms. 

We restrict our attention to two-variable constraints (i.e., of arity 2) 
over boolean domain (i.e. alphabet of size 2).
To simplify matters even further we shall start with the case 
where all the constraints use the same predicate $P:\zo^2\to\zo$.
This restricted case of CSP sparsification 
already generalizes cut-sparsification 
--- simply represent every vertex $v\in V$ by a variable $x_v$, 
and every edge $(v,u)\in E$ by the constraint $x_v\neq x_u$.

Observe that such CSPs capture also other interesting graph problems,
such as the \emph{uncut edges} (using the predicate $x_v = x_u$), 
\emph{covered edges} (using the predicate $x_v \vee x_u$)
or the \emph{directed-cut edges} (using the predicate $x_v\wedge\neg x_u$).
Even though these graph problems are well-known and extensively studied,
we are not aware of any sparsification results for them,
and at a first glance such sparsification may even seem surprising, 
because these problems do not have the combinatorial structure exploited 
by \cite{BK96} (a bound on the number of approximately minimum cuts),
or the linear-algebraic description used by \cite{SS11,BSS14}
(as quadratic forms over Laplacian matrices).

\paragraph{Results.}
For CSPs consisting of a single predicate $P:\zo^2\to\zo$, 
we show in \theoremref{thm:main} that 
a $(1+\epsilon)$-sparsifier of size $O(n/\epsilon^2)$ always exists 
if and only if $\card{P^{-1}(1)} \neq 1$ 
(i.e., $P$ has 0,2,3 or 4 satisfying inputs). 
Observe that the latter condition includes the two graphical examples above
of uncut edges and covered edges, but excludes directed-cut edges.
We further show in \theoremref{thm:or_sketch} that our sparsity bound above
is tight, except for some relatively trivial predicates $P$.
We then build on our sparsification result in Section~\ref{sec:applications}
to obtain $(1+\epsilon)$-sparsifiers for other CSPs, including 
2SAT (which uses 4 predicate types) and 2LIN (which uses 2 predicate types).

Finally, we explore 
future directions, such as more general predicates 
and a generalization of the sparsification paradigm to sketching schemes.
In particular, we see that the above dichotomy  
according to number of satisfying inputs to the predicate extends to sketching.

\section{Two-Variable Boolean Predicates and Digraphs}
\label{sec:predicates}

A \emph{predicate} is a function $\pred{P}:\zo^2\to\zo$ (recall we restrict 
ourselves throughout to two variables and a boolean domain).
Given a set of variables $V$, a \emph{constraint} $\tuple{(v,u),\pred{P}}$ 
consists of a predicate $\pred{P}$ and an ordered pair $(v,u)$ of variables from $V$.
For an assignment $A:V\to \zo$, we say that $A$ \emph{satisfies} the constraint
whenever $\pred{P}(A(v),A(u))=1$.
A VCSP (Valued Constraint Satisfaction Problem) instance $\mathcal{I}$ is a triple $(V,\Pi,w)$, 
where $V$ is a set of variables, $\Pi$ is a set of constraints over $V$ 
(each of the form $\pi_i=\tuple{(v_i,u_i),p_i}$), 
and $w:\Pi\rightarrow\R_+$ is a weight function. 
The \emph{value} of an assignment $A:V\to\zo$ 
is the total weight of the satisfied constraints, i.e., 
\[
  \Val_{\mathcal{I}}\left(A\right)
  :=\sum_{\pi_{i}\in\Pi}w (\pi_{i})\cdot p_{i}(A(v_{i}),A(u_{i})) .
\]
For $\epsilon\in(0,1)$, an \emph{$\epsilon$-sparsifier} of $\mathcal{I}$ 
is a (re-weighted) sub-instance $\mathcal{I}_\epsilon=(V,\Pi_{\epsilon}\subseteq\Pi,w_{\epsilon})$ where
\[
  \forall A:V\to\zo, \qquad
  \Val_{\mathcal{I}_{\epsilon}}(A) \in (1\pm\epsilon)\cdot \Val_{\mathcal{I}}(A) .
\]
The goal is to minimize the number of constraints, i.e., $\card{\Pi_{\epsilon}}$. 
There are $16$ different predicates $\pred{P}:\{0,1\}^{2}\rightarrow\{0,1\}$,
which are listed in \figureref{fig:predicates} with names for easy reference.

	\begin{figure}
		\begin{center}
			\begin{tabular}{|c|c||c|c|c|c|c|c|c|c|c|c|c|c|c|c|c|c|}
				\hline 
				$x_{1}$ & $x_{2}$ & $\vec{0}$ & \pred{nOr} & $01$ & $0x$ & \pred{Dicut} & $x0$ & \pred{Cut} & \pred{nAnd} & \pred{And} & \pred{unCut} & $x1$ & $\overline{10}$ & $1x$ & $\overline{01}$ & \pred{Or} & $\vec{1}$\tabularnewline
				\hline \hline 
				0 & 0 &  & 1 &  & 1 &  & 1 &  & 1 &  & 1 &  & 1 &  & 1 &  & 1\tabularnewline
				\hline 
				0 & 1 &  &  & 1 & 1 &  &  & 1 & 1 &  &  & 1 & 1 &  &  & 1 & 1\tabularnewline
				\hline 
				1 & 0 &  &  &  &  & 1 & 1 & 1 & 1 &  &  &  &  & 1 & 1 & 1 & 1\tabularnewline
				\hline 
				1 & 1 &  &  &  &  &  &  &  &  & 1 & 1 & 1 & 1 & 1 & 1 & 1 & 1\tabularnewline
				\hline 
			\end{tabular}
			
\caption{\label{fig:predicates}
  All possible predicates $\pred{P}:\zo^2\to\zo$, where blank cells denote value $0$. 
  Predicates $0x,x0,x1,1x$ are determined by a single variable. 
  Predicates $01,\pred{Dicut},\overline{10},\overline{01}$ 
  are satisfied by a single assignment or all but a single one.    
}
\hrulefill
		\end{center}
	\end{figure}

We first focus on the case where all the constraints in $\Pi$ 
use the same predicate $\pred{P}$,%
\footnote{The collection of predicates used in a VCSP 
is sometimes called its \emph{signature}. 
In this paper we mainly deal with VCSPs whose signature is of size one.},
in which case we can represent 
the VCSP $\mathcal{I}$ by an edge-weighted digraph $G^{\mathcal{I}}=(V,E,w)$. 
Each variable in $V$ is represented by a vertex, 
and each constraint over the pair $(v,u)$ will be represented 
by a directed edge from $v$ to $u$, with the same weight as the constraint (formally, $E=\aset{ (v,u) \mid (\tuple{v,u},\pred{P}) \in\Pi}$, 
and abusing notation set edge weights $w(v,u)=w(\tuple{(v,u),P})$). 
This transformation preserves all the information about the VCSP 
and allows us to make reductions between VCSPs with different predicates $\pred{P}$ 
as their sole predicate.
	
Given a digraph $G$, a predicate \pred{P} and a subset $S\subseteq V$, define
\[
  \pred{P}_{G}(S)
  :=\sum_{(v,u)\in E}\pred{P}(\indic_{S}(v),\indic_{S}(u))\cdot w((v,u)) ,
\]
where $\indic_{S}$ denotes the indicator function.
For example, applying this definition to the cut predicate 
$\pred{Cut}: (x,y) \to \indic_{\aset{x\ne y}}$, 
we have 
\[
  \pred{Cut}_{G}(S)
  = \sum_{(v,u)\in E}\pred{Cut}(\indic_S(v),\indic_S(u)) \cdot w((v,u))
  = \sum_{(v,u)\in E} \abs{ \indic_S(v) - \indic_S(u) } \cdot  w((v,u)) ,
\]
which is just the total weight of the edges crossing the cut $S$.
This matches the definition we gave in the introduction,
except for the technical subtlety that $G$ is now a directed graph, 
which makes no difference for symmetric predicates like \pred{Cut}.
We shall assume henceforth that $G$ is directed.

We shall say that a sub-instance $G_\epsilon$ is 
an \emph{$\epsilon$-$\pred{P}$-sparsifier} of $G$ if 
\[
  \forall S\subseteq V, \qquad
  \pred{P}_{G_\epsilon}(S) \in (1\pm\epsilon)\cdot \pred{P}_{G}(S) .
\]
Observe that given an assignment $A$ for the variables $V$, 
we can set $S_{A}:=\aset{ u\mid A(u)=1 } $.
It then holds that $\Val_{\mathcal{I}}(A)=\pred{P}_{G^\mathcal{I}}(S_{A})$, 
where $G^\mathcal{I}$ is the appropriate digraph for the VCSP.
As there a bijection between such VCSPs and digraphs, we conclude
\begin{observation}\label{obs: vcsp_graph_similarity}
	The existence of an $\epsilon$-$\pred{P}$-\emph{sparsifier} $G_\epsilon=(V,E_\epsilon,w_\epsilon)$ for $G^\mathcal{I}$ implies the existence of an $\epsilon$-sparsifier $\mathcal{I}_\epsilon$ for $\mathcal{I}$ with $|E_\epsilon|$ constraints.
\end{observation}

Note that the converse is true as well, i.e., an $\epsilon$-sparsifier for $\mathcal{I}$ implies the existence of $\epsilon$-\pred{P}-sparsifier for $G_\mathcal{I}$ of size $|\Pi_\epsilon|$. 
From now on, we focus on finding an $\epsilon$-$\pred{P}$-sparsifier for an arbitrary digraph $G$ (for different choices of the predicate $\pred{P}$).

\section{A Single Predicate}
\label{sec:OnePredicate}
	
In this section we go over all the predicates $\pred{P}:\zo^2\to\zo$ 
and classify them into sparsifiable and non-sparsifiable predicates,
see Theorems~\ref{thm:cut_to_uncut}, \ref{thm:reduction_to_And}, 
and~\ref{thm:main}.
For simplicity, we state our sparsification results as existential, 
but in fact all these sparsifiers can be computed in polynomial time.
Our main technique is a simple graph transformation,
which seems to be very well-known but in other contexts.
We find it surprising that rather different predicates can be analyzed
so easily by applying the same elementary transformation.

In our classification, we appeal to two basic predicates,
the first of which is $\pred{Cut}$, which is already known to be sparsifiable. 
\begin{theorem}[\cite{BSS14}] 
\label{thm:cut}
For every digraph $G$
and parameter $\epsilon\in (0,1)$, there is an  $\epsilon$-$\pred{Cut}$-\emph{sparsifier} for $G$ with $O\left(|V|/\epsilon^{2}\right)$ edges.
\end{theorem}
Our second basic predicate is the predicate $\pred{And}$,
which behaves significantly different. 
We call a digraph $G=(V,E)$ \emph{strongly asymmetric} 
if for every $(v,u)\in E$ it holds that $(u,v)\notin E$.

\begin{theorem}
\label{thm:And}
For every strongly asymmetric digraph $G=(V,E,w)$ 
with strictly positive weights and $\epsilon\in(0,1)$, 
every $\epsilon$-$\pred{And}$-sparsifier $G_{\epsilon}=(V,E_{\epsilon},w_{\epsilon})$
must satisfy $E_{\epsilon}=E$.
\end{theorem}

\begin{proof}
Let $G_{\epsilon}=\left(V,E_{\epsilon},w_{\epsilon}\right)$ be such a sparsifier, i.e., for every $S\subseteq V$ it holds that $\pred{And}_{G_{\epsilon}}(S)\in(1\pm\epsilon)\cdot\pred{And}_{G}(S)$.	
Then for every $e=(v,u)\in E$ we must have $(v,u)\in E_\epsilon$, 
as otherwise for the set $S=\aset{u,v}$ it will hold that $\pred{And}_{G_\epsilon}(\aset{u,v})=0$ while $\pred{And}_{G}(\aset{u,v})=w(e)>0$, 
a contradiction.
\end{proof}
	
\begin{remark}
For every digraph (which is not necessarily strongly asymmetric), 
the same proof shows that $|E_\epsilon|\ge \frac{1}{2}|E|$. 
\end{remark}
\begin{remark}
Our definition of an $\epsilon$-\pred{P}-sparsifier requires $G_\epsilon$ to be a subgraph of $G$, but we can state \theoremref{thm:And} in a more general way: For every digraph $G_{\epsilon}=(V,E_{\epsilon},w_{\epsilon})$ (not necessarily a subgraph) such that every $S\subseteq V$ satisfies 
$\pred{And}_{G_{\epsilon}}(S) \in (1\pm\epsilon)\cdot\pred{And}_{G}(S)$
necessarily $E_{\epsilon}$ agrees with $E$ up to the directions of the edges.
\end{remark}

Next, we show that every other predicates is similar either to \pred{Cut} 
or to $\pred{And}$ in terms of sparsifability. We describe a reduction
that will be useful to show both sparsifability and non-sparsifability.
(This reduction is based on a well-known transformation of a given graph,
called the ``bipartite double cover'', 
see e.g.~\cite{BHM80}, although we are not aware of its use in the same way.)
Let $\gamma$ be a function that maps a digraph $G=\left(V,E,w\right)$ where $V=\left\{ v_{1},v_{2},\dots,v_{n}\right\} $
to a digraph $\gamma(G) = (V^\gamma,E^\gamma,w^\gamma)$
where $V^\gamma = \aset{ v_{-n},\dots,v_{-1},v_{1},\dots,v_{n} }$,
$E^\gamma = \aset{ (v_{i},v_{-j}) \mid (v_{i},v_{j} )\in E }$, $w^\gamma((v_{i},v_{-j})) = w((v_{i},v_{j}))$.
For every subset $S\subseteq V$, 
we introduce the notation $-S:=\aset{ v_{-i}\mid v_{i}\in S } $,
$\bar{S}:=\aset{ v_{i}\mid v_{i}\in V\setminus S }$ and 
$-\bar{S}:=\aset{ v_{-i}\mid v_{i}\in V\setminus S }$.
\figureref{fig:reductions} illustrates the effect of $\gamma$ on 
an arbitrary set $S$.
	
	\begin{figure}
		\begin{center}
			\includegraphics[width=0.4\textwidth]{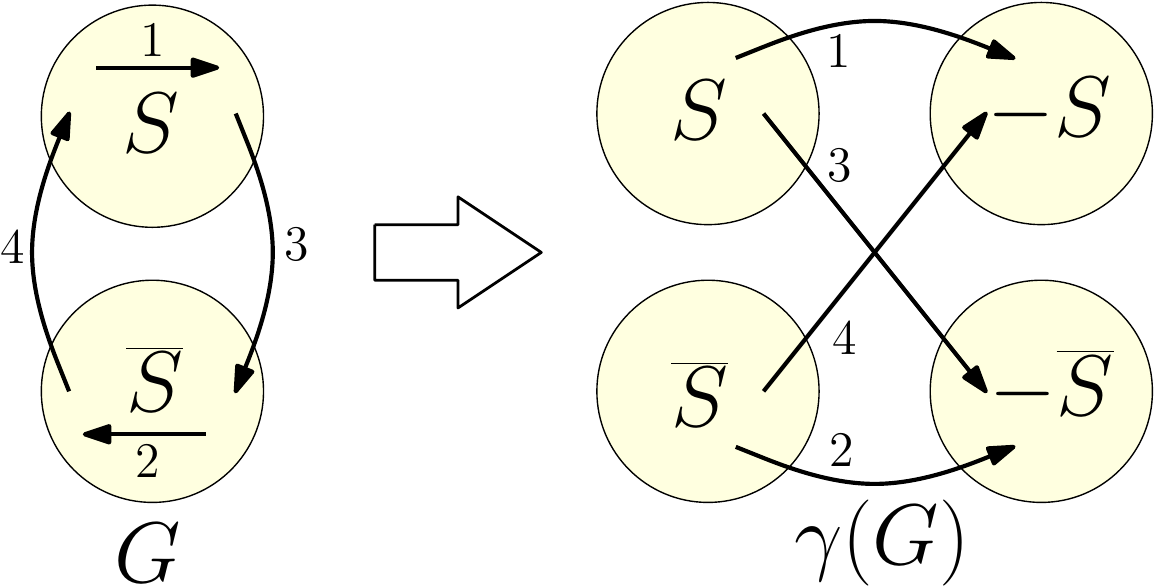}
			\caption{\small  The mapping $\gamma$ applied on $G$ and its effect on an arbitrary $S \subseteq V$. For example, an edge from $v_i\in S$ to $v_j\in \bar{S}$ is represented by an arrow of type 3, and becomes in $\gamma(G)$ an edge from $v_i\in S$ to $v_{-j}\in -\bar{S}$.}
			\label{fig:reductions}
		\end{center}
                \hrule
	\end{figure}

\begin{theorem}
\label{thm:cut_to_uncut}
For every digraph $G=(V,E,w)$ and $\epsilon\in(0,1)$
there is a sub-digraph $G_{\epsilon}$ with $O(|V|/\epsilon^{2})$ edges,
such that for every predicate $\pred{P} \in \aset{ \pred{Cut},\pred{unCut},\pred{Or}, \pred{nAnd},\overline{10}, \overline{01},x0,x1,0x,1x,\vec{1},\vec{0} }$,
the digraph $G_{\epsilon}$ is an $\epsilon$-$\pred{P}$-\emph{sparsifier} of $G$.
(Note that $G_\epsilon$ does not depend on $\pred{P}$.)
\end{theorem}

\begin{proof}
Given $G$ and $\epsilon$, first construct $\gamma(G)$ as above.
Next, apply \theoremref{thm:cut} to obtain for $\gamma(G)$ a cut-sparsifier
$\gamma(G)_{\epsilon}=(V^\gamma,E^\gamma_{\epsilon}\subseteq E^\epsilon,w^\gamma_\epsilon)$,
which contains $O(\card{V^\gamma}/\epsilon^2)=O(\card{V}/\epsilon^2)$ edges.
Now construct a digraph $G_{\epsilon}=(V,E_\epsilon,w_{\epsilon})$ where 
$E_{\epsilon} = \aset{ (v_{i},v_{j}) \mid (v_{i},v_{-j})\in E^\gamma_\epsilon }$
and $w_{\epsilon}(v_{i},v_{j}) = w^\gamma_\epsilon (v_{i},v_{-j})$.
Observe that $\gamma(G_\epsilon)=\gamma(G)_{\epsilon}$,
i.e. if we apply $\gamma$ on $G_\epsilon$ we get exactly $\gamma(G)_\epsilon$.
		
Now suppose that for a predicate $\pred{P}$, 
there is a function $f_P:2^{V} \to 2^{V^\gamma}$ such that 
for every digraph $H$ on the vertex set $V$,
it holds that 
\begin{align}  \label{eq:PHS}
  \forall S\subset V, \qquad
  \pred{P}_{H}(S) = \pred{Cut}_{\gamma(H)}(f_P(S)).  
\end{align}
Then we could apply~\eqref{eq:PHS} twice, 
first to $G_\epsilon$ and then to $G$, and obtain that
\[
  \forall S\subset V, \qquad
  \pred{P}_{G_{\epsilon}}(S)
  = \pred{Cut}_{\gamma(G)_{\epsilon}}(f_P(S)) 
  \in (1\pm\epsilon) \cdot\pred{Cut}_{\gamma(G)} (f_P(S))
  = (1\pm\epsilon)\cdot \pred{P}_{G}(S) .
 \]
Hence, the existence of such a function $f_P$ 
implies that $G_\epsilon$ is an $\epsilon$-$\pred{P}$-sparsifier. 
And indeed, we can show such $f_P$ for some predicates $\pred{P}$, as follows.
\begin{itemize} \compactify
\item $f_{\pred{unCut}}(S)=S\cup \bar{S}$; 
\item $f_{\pred{Cut}}(S)=S\cup -S$;
\item $f_{0x}(S)=\bar{S}$;
\item $f_{x0}(S)=-\bar{S}$;
\item $f_{x1}(S)=-S$;
\item $f_{1x}(S)=S$;
\item $f_{\vec{1}}(S)=S\cup\bar{S}$; and 
\item $f_{\vec{0}}(S)=\emptyset$. 
\end{itemize}
To verify that $f_{\pred{unCut}}(S)=S\cup \bar{S}$ satisfies \eqnref{eq:PHS}, i.e., that $\pred{unCut}_H(S) = \pred{Cut}_{\gamma(H)}(S\cup \bar S)$,
observe that both sides consist exactly of the edges of types $1$ and $2$
in \figureref{fig:reductions}.
The other predicates can be easily verified similarly, 
which completes the proof for all 
$\pred{P}\in \aset{ \pred{Cut}, \pred{unCut}, 0x, x0, x1, 1x, \vec{1}, \vec{0} }$.

To show that $G_\epsilon$ is a sparsifier also for predicates 
$\pred{P}\in \aset{ \pred{Or}, \pred{nAnd}, \overline{10},\overline{01}}$ 
we need a slightly more general argument. 
Suppose that for a predicate $\pred{P}$, 
there are functions $f^1_P,f^2_P,f^3_P:2^{V}\to 2^{V^\gamma}$ 
such that for every digraph $H$ on the vertex set $V$, 
\begin{align}  \label{eq:PHS3}
  \pred{P}_{H}(S)
  = \tfrac12 \left[\pred{Cut}_{\gamma(H)}(f_P^1(S))
  +\pred{Cut}_{\gamma(H)}(f_{P}^2(S))
  +\pred{Cut}_{\gamma(H)}(f_{P}^3(S))
  \right] .
\end{align}
Then we could apply~\eqref{eq:PHS3} twice,
first to $G_\epsilon$ and then to $G$, and obtain that
\begin{eqnarray*}
  \pred{P}_{G_{\epsilon}}\left(S\right)
  &=&\tfrac12 \left[
      \pred{Cut}_{\gamma(G)_{\epsilon}}(f_{P}^{1}(S))
      +\pred{Cut}_{\gamma(G)_{\epsilon}}(f_{P}^{2}(S))
      +\pred{Cut}_{\gamma(G)_{\epsilon}}(f_{P}^{3}(S)) \right]
  \\
  &\in& (1\pm\epsilon)\cdot\tfrac12\left[
      \pred{Cut}_{\gamma(G)}(f_{P}^{1}(S))
      +\pred{Cut}_{\gamma(G)}(f_{P}^{2}(S))
      +\pred{Cut}_{\gamma(G)}(f_{P}^{3}(S)) \right]
  \\
  &=& (1\pm\epsilon) \cdot \pred{P}_{G}(S) .
\end{eqnarray*}
Hence, the existence of such three functions will imply that $G_\epsilon$ 
is an $\epsilon$-$\pred{P}$-sparsifier. 
And indeed, we let 
\begin{itemize} \compactify
\item $f_{\pred{Or}}^{1}(S)=S$, 
$f_{\pred{Or}}^{2}(S)=-S$, 
$f_{\pred{Or}}^{3}(S)=S\cup-S$;
\item $f_{\pred{nAnd}}^{1}(S)=\bar{S}$, 
$f_{\pred{nAnd}}^{2}(S)=-\bar{S}$, 
$f_{\pred{nAnd}}^{3}(S)=\bar{S}\cup-\bar{S}$;
\item $f_{\overline{10}}^{1}(S)=\bar{S}$, 
$f_{\overline{10}}^{2}(S)=-S$, 
$f_{\overline{10}}^{3}(S)=\bar{S}\cup-S$; and
\item
$f_{\overline{01}}^{1}(S)=S$, 
$f_{\overline{01}}^{2}(S)=-\bar{S}$, 
$f_{\overline{01}}^{3}(S)=S\cup-\bar{S}$.
\end{itemize}		
To verify that $f^1_{\pred{Or}},f^2_{\pred{Or}},f^3_{\pred{Or}}$ satisfies \eqnref{eq:PHS3}, 
observe that both sides consist exactly of the edges of types $1,3,4$ 
in \figureref{fig:reductions}.
The other predicates can be easily verified similarly, 
which completes the proof for all 
$\pred{P}\in \aset{ \pred{Or},\pred{nAnd},\overline{10},\overline{01} }$.
\end{proof}

	Next, we use $\gamma$ for a reductions from $\pred{And}$ to all the remaining predicates. In particular it will imply their ``resistance to sparsification''.
	
	\begin{theorem}
		\label{thm:reduction_to_And} Given parameters $n$ and $m\le {n \choose 2}$, there is a digraph $G=\left(V,E,w\right)$ with $2n$ vertices and $m$ edges such that for every $\epsilon\in(0,1)$ and every predicate $\pred{P}\in\left\{\pred{nOr},01,\pred{Dicut}, \pred{And}\right\}$, for every $\epsilon$-$\pred{P}$-sparsifier $G_\epsilon=(V,E_\epsilon,w_\epsilon)$ of $G$  it holds that that $E_{\epsilon}=E$. (Note that $G$ does not depend on $\pred{P}$.)
	\end{theorem}	
	
	\begin{proof}
		Let $G=\left(V,E,w\right)$ be an arbitrary strongly asymmetric digraph with $n$ vertices, $m$ edges and strictly positive weights. Let $\gamma(G)$ be the digraph constructed by our reduction. Note that $\gamma(G)$ consist of $2n$ vertices and $m$ edges. $\gamma(G)$ will be the digraph for which we will prove the theorem.
		
		Fix some predicate \pred{P}. Let $\gamma(G)_{\epsilon}=\left(V^\gamma,E^\gamma_{\epsilon}\subseteq E^\epsilon,w^\gamma_{\epsilon}\right)$ be some $\epsilon$-$\pred{P}$-sparsifier for $\gamma(G)$.
		Let $G_{\epsilon}=\left(V,E_{\epsilon},w_{\epsilon}\right)$ be a
		digraph where $E_{\epsilon}=\left\{ \left(v_{i},v_{j}\right)\mid\left(v_{i},v_{-j}\right)\in E^\gamma_{\epsilon}\right\} $
		and $w_{\epsilon}\left(\left(v_{i},v_{j}\right)\right)=w^\gamma_{\epsilon}\left(\left(v_{i},v_{-j}\right)\right)$.
		Note that $\gamma(G_\epsilon)=\gamma(G)_{\epsilon}$. 
		
		Now suppose that there is a function $f_P:2^{V} \to 2^{V^\gamma}$ such that 
		for every digraph $H$ on the vertex set $V$,
		it holds that 
		\begin{align}  \label{eq:PHS_And}
		\forall S\subset V, \qquad
		\pred{And}_{H}\left(S\right)=\pred{P}_{\gamma(H)}\left(f_P(S)\right).  
		\end{align}
		Then we could apply~\eqref{eq:PHS_And} twice, 
		first to $G_\epsilon$ and then to $G$, and obtain that
		\[
		\forall S\subset V, \qquad
		\pred{And}_{G_{\epsilon}}(S)
		= \pred{P}_{\gamma(G)_{\epsilon}}(f_P(S)) 
		\in (1\pm\epsilon) \cdot\pred{P}_{\gamma(G)} (f_P(S))
		= (1\pm\epsilon)\cdot \pred{And}_{G}(S) .
		\]
		Hence, assuming such a function $f$ exists, $G_\epsilon$ is an $\epsilon$-\pred{And}-sparsifier for $G$. According to \theoremref{thm:And}, necessarily $E_\epsilon=E$, and in particular $E^\gamma_{\epsilon}=E^\gamma$.
		
		Hence, The existence of such functions $f_P$ for all $\pred{P}\in\left\{\pred{nOr},01,\pred{Dicut}, \pred{And}\right\}$ will imply our theorem.
		And indeed, we let 
		\begin{itemize} \compactify
			\item $f_{And}(S)=S\cup-S$;
			\item $f_{nOr}(S)=\bar{S}\cup -\bar{S}$;
			\item $f_{Dicut}(S)=S\cup-\bar{S}$; and
			\item $f_{01}(S)=\bar{S}\cup-S$.
		\end{itemize}	
		
		To verify that $f_{Dicut}(S)=S\cup-\bar{S}$ satisfies \eqnref{eq:PHS_And}, 
		observe that both sides consist exactly of the edges of type $1$ 
		in \figureref{fig:reductions}.
		The other predicates can be easily verified similarly. 
	\end{proof}
	
We conclude our main theorem, which basically puts together Theorems~\ref{thm:cut_to_uncut} and~\ref{thm:reduction_to_And}.

	\begin{theorem}\label{thm:main}
		Let $\pred{P}$ be a binary predicate, and let $\epsilon\in(0,1)$ be some parameter. 
		\begin{itemize}
			\item If $\pred{P}$ has a single ``1'' in its truth table then there exist a VCSP $\mathcal{I}=\left(V,\Pi,w\right)$ with a single predicate $\pred{P}$, such that every $\epsilon$-$\pred{P}$-sparsifier of $\mathcal{I}$ will have $\Omega(|V|^{2})$ constraints. 
			\item If $\pred{P}$ does not has a single ``1'' in its truth table then for every VCSP $\mathcal{I}=\left(V,\Pi,w\right)$ with single predicate $\pred{P}$, there exists an $\epsilon$-$\pred{P}$-sparsifier
			with $O\left(|V|/\epsilon^{2}\right)$ constraints.
		\end{itemize}
	\end{theorem}

\section{Lower Bounds (for a Single Predicate)}
\label{sec:lwr_bound}

In this section we will show that \theoremref{thm:cut_to_uncut} is tight. More precisely, we will show that for every $\pred{P} \in \aset{ \pred{Cut},\pred{unCut},\pred{Or}, \pred{nAnd},\overline{10},\overline{01}}$,
there exists an $n$-vertex graph $G$ such that every $\epsilon$-\pred{P}-sparsifier $G_\epsilon$ of $G$ must contain $\Omega(n/\epsilon^2)$ edges.%
\footnote{The other predicates $\aset{ x0,x1,0x,1x,\vec{1},\vec{0} }$, are kind of trivial in the sense of sparsification. $\vec{0}$ sparsified by the empty graph. $\vec{1}$ can be sparsified using a single edge. $\aset{ x0,x1,0x,1x}$ could be sparsified using $n$ edges.
}
The first step was done by \cite{ACKQWZ15}, who showed that \theoremref{thm:cut} is tight, i.e., for every $n$ and $\epsilon\in (1/\sqrt{n},1)$, there exists $n$-vertex graph $G$ such that  every $\epsilon$-\pred{Cut}-sparsifier $G_\epsilon$ of $G$ must contain $\Omega(n/\epsilon^2)$ edges.
Using our reduction $\gamma$ in similar manner to \theoremref{thm:cut_to_uncut}, this lower bound can be extended to \pred{unCut} based on the fact that $\pred{Cut}_G(S)=\pred{unCut}_{\gamma(G)}\left(S\cup -\bar{S}\right)$. However, $\gamma$ fails to extend the lower bound to predicates with three $1$'s in their truth table. To this end, we will define sketching schemes, a variation of sparsification
where the goal is to maintain the approximate value of every assignment 
using a small data structure, possibly without any combinatorial structure,
see definition below.
We will use a lower bound on the sketch size of \pred{Cut} from \cite{ACKQWZ15} to prove lower bound on the number of edges in a sparsifier (and also on the sketch size) for \pred{OR}. The extension to other predicates with three $1$'s in their truth table is straightforward using $\gamma$.
Sketching is interesting for its own, and we have further discussion and lower bounds regarding sketching in Section~\ref{sec:sketching}. 

\medskip
Formally, a \emph{sketching scheme} (or a \emph{sketch} in short) is a pair of algorithms $(\sk,\est)$.
Given a weighted digraph $G=(V,E,w)$ and a predicate \pred{P}, 
algorithm $\sk$ returns a string $\sk_{G}$ 
(intuitively, a short encoding of the instance). 
Given $\sk_{\mathcal{I}}$ and a subset $S\subseteq V$, 
algorithm $\est$ returns a value (without looking at $G$) 
that estimates $\pred{P}_{G}(S)$.	
We say that it is an $\epsilon$-\pred{P}-\emph{sketching-scheme} if for every digraph $G$, and for every subset $S\subseteq V$, $\est(\sk_G,S)\in (1\pm \epsilon)\cdot \pred{P}_{G}(S)$.
The \emph{sketch-size} is $\max_{G}|\sk_G|$, 
the maximal length of the encoding string over all the digraphs with $n$ variables, often measured in bits.	
$sk$ might be probabilistic algorithm, but for our purposes it is enough to think only on the deterministic case. 
Note that an algorithm for constructing $\epsilon$-sparsifiers always provides an $\epsilon$-sketching-scheme, where the sketch-size is asymptotically equal to the number of constraints in the constructed sparsifiers when measured in machine words (and up to logarithmic factors when measured in bits).
Sparsification is advantageous over general sketching as it preserves the combinatorial structure of the problem. Nevertheless, one may be interested in constructing sketches as they may potentially require significantly smaller storage.

\begin{theorem}\label{thm:or_sketch}
	Fix a predicate $\pred{P} \in \aset{ \pred{Cut},\pred{unCut},\pred{Or}, \pred{nAnd},\overline{10}}$, an integer $n$ and $\epsilon\in(1/\sqrt{n},1)$. 
	The sketch-size of every $\epsilon$-\pred{P}-sketching-scheme on $n$ variables is $\Omega(n/\epsilon^2)$.
	Moreover, there is an $n$-vertex digraph $G$, such that every $\epsilon$-\pred{P}-sparsifier of $G$ has $\Omega(n/\epsilon^2)$ edges.
\end{theorem}		

\begin{proof}
	We follow the line-of-proof of Theorems 4.1 and 4.2 in \cite{ACKQWZ15}.
	Specifically, they show that the sketch-size of every $\epsilon$-\pred{Cut}-sketching-scheme is $\Omega(n/\epsilon^2)$ bits,
	by proving that a certain family $\mathcal{F}$ of $n$-vertex graphs 
	is hard to sketch, and consequently to sparsify. 
By similar arguments to \theoremref{thm:cut_to_uncut}, 
this lower bound easily extends to \pred{unCut}. 
Indeed, recall that $\pred{Cut}_G(S)=\pred{unCut}_{\gamma(G)}\left(S\cup -\bar{S}\right)$, and thus a $\epsilon$-\pred{unCut}-sparsifier (or sketch) for $\gamma(G)$ yields an $\epsilon$-\pred{Cut}-sparsifier (or sketch) 
for $G$ with the same number of edges (size).
	
Once we prove the lower bound for predicate \pred{OR},
a reduction from \pred{OR} using $\gamma$ will extend it also to \pred{nAnd}, $\overline{10}$ and $\overline{01}$, because
	\begin{equation}\label{eq:or_by_nand}
\pred{Or}_{G}(S)=\pred{nAnd}_{\gamma(G)}(\bar{S}\cup-\bar{S})=\overline{01}_{\gamma(G)}(S\cup-\bar{S})=\overline{10}_{\gamma(G)}(\bar{S}\cup-S).	
	\end{equation} 	
We will thus focus on the predicate \pred{OR}. 
As it is symmetric predicate, we can work with graphs rather then digraphs. 
The main observation in our proof is that for every undirected graph $G=(V,E,w)$,
	if $\deg_G(v)$ denotes the degree of vertex $v$, then 
	\begin{equation} \label{eq:Or2Cut}
	\forall S\subset V, \qquad 
	\pred{Cut}_G(S) = 2\cdot\pred{OR}_G(S) - \sum_{v\in S}\deg_G(v).
	\end{equation}
	
	The graph family $\mathcal{F}$ consists of graphs $G$ constructed as follows.
	Let $s_1,\dots,s_{n/2}\in \{0,1\}^{1/\epsilon^2}$ be balanced $1/\epsilon^2$ bit-strings (i.e., each $s_i$ has normalized Hamming weight exactly $1/2$), 
	and let the graph $G$ be a disjoint union of the graphs 
	$\{G_{j}\mid j\in[\epsilon^{2}n/2]\}$,
	where each $G_j$ is a bipartite graph, whose two sides, 
	each of size $1/\epsilon^2$, are denoted $L(G_j)$ and $R(G_j)$.
	The edges of $G$ are determined by $s_1,\dots,s_{n/2}$, 
	where each bit string $s_i$ is indicates the adjacency between 
	vertex $i\in\cup_j L(G_j)$ and the vertices in the respective $R(G_j)$.
	They further observe (in Theorem 4.2) that the lower bound holds even if
	the sketching scheme is relaxed as follows:
	\begin{enumerate}
		\item \label{itm:cut_types} 
		The estimation is required only for cut queries contained in a single $G_j$,
		namely, cut queries $S\cup T$ where $S\subset L(G_j)$ and $T\subset R(G_j)$ 
		for the same $j$.
		\item \label{itm:additive_err} 
		The estimation achieves additive error $\mu/\epsilon^3$, where $\mu=10^{-4}$
		(instead of multiplicative error $1\pm \epsilon$).
	\end{enumerate}
	
To prove a sketch-size lower bound for a 
$(\mu\epsilon)$-\pred{OR}-sketching-scheme $(\sk^{\pred{OR}},\est^{\pred{OR}})$,
we assume it has sketch-size $s=s(n,\epsilon)$ bits, 
and use it to construct a
\pred{Cut}-sketching-scheme $(\sk^{\pred{Cut}},\est^{\pred{Cut}})$ that achieves 
	the estimation properties \ref{itm:cut_types} and \ref{itm:additive_err} 
	on graphs of the aforementioned form, 
	and has sketch-size $s+2n\log (1/\epsilon)$ bits.
	Then by \cite{ACKQWZ15}, this sketch-size must be $\Omega(n/\epsilon^2)$,
	and we conclude that $s=\Omega(n/\epsilon^2)$ as required.
	
	Given a graph $G\in\mathcal{F}$, 
	let $\sk^{\pred{Cut}}_G$ be a concatenation of $\sk^{\pred{OR}}_G$ 
	and a list of all vertex degrees in $G$. 
	The degrees in $G$ are bounded by $1/\epsilon^2$, 
	hence the size of $\sk^{\pred{Cut}}_G$ is indeed $s+2n\log (1/\epsilon)$ bits. 
	Given a cut query $S\cup T$ contained in some $G_j$,
	define the estimation algorithm (which we now construct for \pred{Cut}) to be
	\begin{equation} \label{eq:est_cut}
	\est^{\pred{Cut}}(\sk^{\pred{Cut}}_G,S\cup T ) 
	\eqdef 2\cdot \est^{\pred{OR}}(\sk^{\pred{OR}}_G,S\cup T )
	- \sum_{v\in S\cup T}\deg_G(v) .
	\end{equation}
	
	Let us analyze the error of this estimate. 
	First, observe that as in each $G_j$ there are precisely $ \frac{1}{2\epsilon^4}$ edges, $\pred{OR}_G(S\cup T)\le \frac{1}{2\epsilon^4}$, 
	and thus 
	\[
	\est^{\pred{OR}}(\sk^{\pred{OR}}_G,S\cup T )
	\in (1\pm \mu\epsilon)\cdot \pred{OR}_G(S\cup T)
	\subseteq \pred{OR}_G(S\cup T) \pm \frac{\mu}{2\epsilon^3}~.
	\]
	Plugging this estimate into \eqref{eq:est_cut} and then 
	recalling our initial observation \eqref{eq:Or2Cut}, we obtain as desired
	\begin{align*}
	\est^{\pred{Cut}}(\sk^{\pred{Cut}}_G,S\cup T ) 
	&\in 2\cdot \pred{OR}_G(S\cup T ) \pm \frac{\mu}{\epsilon^3}
	- \sum_{v\in S\cup T}\deg_G(v) \\
	&= \pred{Cut}_G(S\cup T) \pm \frac{\mu}{\epsilon^3} ~.
	\end{align*}
	
	To prove a lower bound on the size of an \pred{OR}-sparsifier, 
	we follow the argument in \cite[Theorem 4.2]{ACKQWZ15}, 
	which shows that given an $\epsilon$-\pred{Cut}-sparsifier $G_\epsilon$ 
	with $s=s(n,\epsilon)$ edges for a graph $G\in \mathcal{F}$, 
	there is a $\pred{Cut}$-sparsifier $G_\mu$ of $G_\epsilon$, 
	with additive error $\mu/2\epsilon^3$,
	such that $G_\mu$ has only integer weights and henceforth 
	can be encoded using $O(s(\mu^{-2}+\log(\epsilon^{-2}n/s)))$ bits.
	In fact, there is nothing special here about \pred{Cut}. 
	The same proof will work (with the same properties) for predicate $\pred{OR}$, assuming a sparsifier is required to be a subgraph 
	(to remove this restriction, just erase all the edges between $G_j$ to $G_i$ for $i\ne j$, which adds only a small additive error).
	
	Now suppose that every graph $G$ of the form specified above admits a $\frac{\mu}{2}\epsilon$-\pred{OR}-sparsifier $G_\epsilon$ with $s$ edges. 
	Then as explained above (about repeating the argument of \cite{ACKQWZ15})
	there is a graph $G_\mu$ that sparsifies $G_\epsilon$ with additive error $\mu/2\epsilon^3$, and can be encoded by a string $\mathcal{I}_G$ of size $O(s\log(\epsilon^{-2}n/s))$ bits (recall that $\mu$ is a constant).
	Use it to construct a \pred{Cut}-sketching-scheme with additive error $\mu/\epsilon^3$ as follows. 
	Given the graph $G$, set $\sk^{\pred{Cut}}_G$ to be the concatenation of $\mathcal{I}_G$ and a list of the degrees of all the vertices in $G$. 
	Then 
	$\card{\mathcal{I}_G} = O(s \log(\epsilon^{-2}n/s)) + 2n\log(1/\epsilon)$.
	For a cut query $S\cup T$ contained in some $G_j$, 
	define the estimation algorithm (using the \pred{OR} sparsifier) to be
	\begin{align*}
	\est^{\pred{Cut}}(\sk^{\pred{Cut}}_G,S\cup T ) 
	&\eqdef 2\cdot \pred{OR}_{G_\mu}(S\cup T )-\sum_{v\in S\cup T}\deg_{G}(v) .
	\intertext{Then we can again analyze it by plugging the above error bounds 
		and then using \eqref{eq:Or2Cut},}
	\est^{\pred{Cut}}(\sk^{\pred{Cut}}_G,S\cup T ) 
	&\in 2\cdot \pred{OR}_{G_\epsilon}(S\cup T ) \pm \frac{\mu}{2\epsilon^3}
	- \sum_{v\in S\cup T}\deg_G(v)
	\\
	&\in 2\cdot \pred{OR}_G(S\cup T ) \pm \frac{\mu}{\epsilon^3} 
	- \sum_{v\in S\cup T}\deg_G(v)
	\\
	&= \pred{Cut}_G(S\cup T) \pm \frac{\mu}{\epsilon^3} ~.
	\end{align*}
	By \cite{ACKQWZ15}, the sketch-size must be 
	$|\mathcal{I}_G|=\Omega(n/\epsilon^2)$,
	hence $s=\Omega(n/\epsilon^2)$ (for at least one graph $G\in\mathcal{F}$) as required. 
\end{proof}

\section{Multiple Predicates and Applications}
\label{sec:applications}

In this section we extend \theoremref{thm:cut_to_uncut} to VCSPs using
multiple types of predicates. 
In particular, we prove sparsifability for some classical problems.
Again, our sparsification results are stated as existential bounds, 
but these sparsifiers can actually be computed in polynomial time.

\begin{theorem}\label{thm: genreal CSP sparsification}
For every $\epsilon\in (0,1)$ and a VCSP $(V,\Pi,w)$ whose constraints 
$\langle \left(v,u\right),\pred{P}\rangle\in \Pi$ all satisfy $\pred{P}\notin\left\{ \pred{nOr},01, \pred{Dicut}, \pred{And}\right\}$, 
there exists an $\epsilon$-sparsifier for $\mathcal{I}$ with $O(|V|/\epsilon^{2})$ constraints.
\end{theorem}

This bound 
is tight, according to \theoremref{thm:or_sketch}.
We prove it by a straightforward application of \theoremref{thm:cut_to_uncut}.
Partition $\mathcal{I}$ to disjoint VCSPs according to the predicates in the constraints, and then for each sub-VCSP find an $\epsilon$-sparsifier using \theoremref{thm:cut_to_uncut}. The union of this sparsifiers is an $\epsilon$-sparsifier for $\mathcal{I}$. 
A formal proof follows.

\begin{proof}[Proof of \theoremref{thm: genreal CSP sparsification}]
		For each predicate $\pred{P}$, let $\Pi^{P}=\left\{ \pi\in\Pi\mid\pi=\left\langle \left(v,u\right),\pred{P}\right\rangle \right\}$. Note that $\{\Pi^P\}$ forms a partition of $\Pi$. For each $\pred{P}$, let  $\mathcal{I}^P=(V,\Pi^P,w^P)$ where $w^P$ is the restriction of $w$ to $\Pi^P$. Let  $\mathcal{I}^P_\epsilon=(V,\Pi^P_\epsilon,w^P_\epsilon)$ be an $\epsilon$-$\pred{P}$-sparsifier for $\mathcal{I}^P$ with $|\Pi^P_\epsilon|=O(|V|/\epsilon^2)$ constraints according to \theoremref{thm:cut_to_uncut} (recall that $\pred{P}\notin\left\{ \pred{nOr},01, \pred{Dicut}, \pred{And}\right\}$). Set $\mathcal{I}_\epsilon=(V,\Pi_\epsilon,w_\epsilon)$, $\Pi_\epsilon=\bigcup_P\Pi^P_\epsilon$ and $w_\epsilon=\bigcup_P w^P_\epsilon$. For every assignment $A$, 
		\begin{eqnarray*}
			\mbox{Val}_{\mathcal{I}_{\epsilon}}(A)
			&=&\sum_{\pi_{i}\in\Pi_{\epsilon}}w_{\epsilon}\left(\pi_{i}\right)\cdot p_{i}\left(A(v_{i}),A(u_{i})\right)\\
			&=&\sum_{\pred{P}}\sum_{\pi_{i}\in\Pi_{\epsilon}^{P}}w_{\epsilon}^{P}\left(\pi_{i}\right)\cdot \pred{P}\left(A(v_{i}),A(u_{i})\right)\\
			&\in&\left(1\pm\epsilon\right)\cdot\sum_{P}\sum_{\pi_{i}\in\Pi^{P}}w^{P}\left(\pi_{i}\right)\cdot \pred{P}\left(A(v_{i}),A(u_{i})\right)\\
			&=&\left(1\pm\epsilon\right)\cdot\sum_{\pi_{i}\in\Pi}w\left(\pi_{i}\right)\cdot p_{i}\left(A(v_{i}),A(u_{i})\right)\\
			&=&\left(1\pm\epsilon\right)\cdot\mbox{Val}_{\mathcal{I}}(A),
		\end{eqnarray*}
and note that indeed $|\Pi_\epsilon|\le O\left(n/\epsilon^{2}\right)$.
	\end{proof}

\pred{2SAT} (boolean satisfiability problem over constraints with 2 variables) can be viewed as a VCSP which uses only the predicates $\pred{Or}$, $\pred{nAnd}$, $\overline{10}$ and $\overline{01}$. 
By \theoremref{thm: genreal CSP sparsification}, for every \pred{2SAT} formula $\Phi$ over $n$ variables, and for every $\epsilon\in(0,1)$, there is a sub-formula $\Phi_\epsilon$ with $O(n/\epsilon^2)$ clauses, such that $\Phi$ and $\Phi_\epsilon$ have the same value for every assignment up to factor $1+\epsilon$.\footnote{We use here the version of \pred{2SAT} where each clause has weight and every assignment has value rather then the version when we only ask weather there an assignment that satisfies all the clauses.}

\pred{2LIN} is a system of linear equations (modulo 2),
where each equation contains 2 variables and has a nonnegative weight. 
Notice that the equation $x+y=1$ is a constraint using the \pred{Cut} predicate while the equation $x+y=0$ is a constraint using the \pred{unCut} predicate. 
By \theoremref{thm: genreal CSP sparsification}, if $n$ denotes the number of
variables, then for every $\epsilon\in(0,1)$ we can construct a sparsifier with only $O(n/\epsilon^2)$ equations 
(i.e., a re-weighted subset of equations, such that on every assignment it agrees with the original system up to factor $1+\epsilon$).

We note that by our lower bound (\theoremref{thm:or_sketch}), 
there are instances of \pred{2SAT} (\pred{2LIN}) for which every $\epsilon$-sparsifier must contain $\Omega(n/\epsilon^2)$ clauses (equations).

\section{Further Directions}
\label{sec:discussion}

Based on the past experience of cut sparsification in graphs
-- which has been extremely successful in terms of techniques, applications, extensions and mathematical connections --
we expect VCSP sparsification to have many benefits.
A challenging direction is to identify which predicates admit sparsification,
and our results make the first strides in this direction.

We now discuss potential extensions to our results in the previous sections
(which characterize two-variable predicates over a boolean alphabet).
We first consider predicates with more variables, 
and in particular show sparsification for $k$-\pred{SAT} formulas, 
in Section~\ref{sec:MoreVars}.
We then consider predicates with large alphabets in Section~\ref{sec:LargeAlphabet},
showing in particular a sparsifier construction for $k$-\pred{Cut}, 
and that linear equations (modulo $k\ge3$) are not sparsifiable.
We also consider sketching schemes, notable we discuss a more loose sketching model called \emph{for-each} in Section~\ref{sec:sketching}. 
Finally, we study \emph{spectral} sparsification for \pred{unCut}, 
a notion that preserves some algebraic properties in addition to the ``uncuts''
in Section~\ref{sec:spectral}.

\subsection{Predicates over more variables and $k$-\pred{SAT}}
\label{sec:MoreVars}
It is natural to ask for the best bounds on the size of $\epsilon$-\pred{P}-sparsifiers for different predicates $\pred{P}:\{0,1\}^k\rightarrow \{0,1\}$.
	A first step towards answering this question was already done by \cite{KK15}.
	\begin{theorem}[\cite{KK15}]\label{thm: cut hypergraph}
		For every hypergraph H = (V,E,w) with hyperedges containing at most $r$ vertices, and $\epsilon\in (0,1)$, there is a re-wighted subhypergraph $H_\epsilon$ with $O(n(r+\log n)/\epsilon^2)$ hyperedges such that
		\[
			\forall S\subseteq V,~~~\pred{Cut}_{H{_\epsilon}}(S)\in (1\pm\epsilon)\cdot \pred{Cut}_H(S) .
		\]
	\end{theorem}    
	Here we say that a hyperedge $e$ is \emph{cut} by $S$ if $S\cap e\notin \{\emptyset,e\}$ (i.e., not all the vertices in $e$ are in the same side). Observe that \pred{Cut} is equivalent to the predicate \pred {NAE} (not all equal). In particular \theoremref{thm: cut hypergraph} implies that for every VCSP using only \pred{NAE}, there is an $\epsilon$-sparsifier with $O(n(r+\log n)/\epsilon^2)$ constraints.
	
A $k$-\pred{SAT} is essentially a VCSP that uses only predicates with a single $0$ in their truth table. 
\cite{KK15} use \theoremref{thm: cut hypergraph} to construct an $\epsilon$-sketching-scheme with sketch-size $\tilde{O}(nk/\epsilon^2)$ for $k$-SAT formulas (i.e., only for VCSPs of this particular form).
We observe that their sketching scheme can be further used to construct an $\epsilon$-sparsfiers, as follows.

First, recall how the sketching scheme of \cite{KK15} works.
Given a $k$-\pred{SAT} formula $\Phi=(V,\mathcal{C},w)$ (variables, clauses, weight over $\mathcal{C}$), construct a hypergraph $H$ 
on vertex set $V\cup -V\cup \{f\}$. 
We associate the literal $v_i$ with vertex $v_i$, the literal $\neg v_i$ with vertex $v_{-i}$, and use $f$ to represent the ``false''. 
Each clause becomes a hyperedge consisting of $f$ and (the vertices associated with) the literals in $\mathcal{C}$ 
(for example $v_5\vee \neg v_7 \vee v_{12}$ becomes $\{f,v_5,v_{-7},v_{12}\}$).
Observe that given a truth assignment $A:V\to\zo$, 
if we define $S_{A}:=\aset{ u\mid A(u)=0 }$, 
then $\Val_\Phi(A)=\pred{Cut}_{H}(S_{A}\cup\{f\})$,
and using \theoremref{thm: cut hypergraph} this provides a sketching scheme.
Moreover, given an $\epsilon$-\pred{Cut}-sparsifier $H_\epsilon$ for $H$, let $\Phi_\epsilon$ be the formula which has only the clauses associated with edges that ``survived'' the sparsification, with the same weight. 
Notice that for every assignment $A$,
\[
  \Val_{\Phi_\epsilon}(A)
  =\pred{Cut}_{H_\epsilon}(S_{A}\cup\{f\})
  \in(1\pm\epsilon) \cdot \pred{Cut}_{H}(S_{A}\cup\{f\})
  = (1\pm\epsilon)\cdot\Val_\Phi(A)~.
\]
\begin{theorem}
Given $k$-\pred{SAT} formula $\Phi$ over $n$ variables and parameter $\epsilon\in(0,1)$, there is an $\epsilon$-sparsifier sub-formula $\phi_\epsilon$ with $O(n(k+\log n)/\epsilon^2)$ clauses.
\end{theorem}

In contrast, we are not aware of any nontrivial sparsification result
for the parity predicate (on $k\ge3$ boolean variables),
and this remains an interesting open problem.

\subsection{Predicates over larger Alphabets}
\label{sec:LargeAlphabet}

Our results deal only with predicates that get two input values in $\{0,1\}$. 
A natural generalization is to sparsify a VCSP that uses a predicate 
over an alphabet of size $k$,
i.e., $\pred{P}:[k]\times [k]\rightarrow \{0,1\}$, 
where $[k] \eqdef \{0,1,\dots,k-1\}$.
One predicate that we can easily sparsify is \pred{NE} (not-equal), which is satisfied if the two constrained variables have are assigned different values.
Indeed, in the graphs language, this is called a $\pred{k-Cut}$, 
where the value of a partition $(S_0,\dots,S_{k-1})$ of the vertices is the total weight of all edges with endpoints in different parts. 
It turns out that $\epsilon$-$\pred{Cut}$-sparsifier is in particular an $\epsilon$-$\pred{k-Cut}$-sparsifier, using the following well-known 
double-counting argument:
\begin{eqnarray*}
\pred{k-Cut}_{G_{\epsilon}}\left(S_{0},\dots,S_{k-1}\right)
&=&\frac{1}{2}\cdot\left[\pred{Cut}_{G_{\epsilon}}\left(S_{0},\overline{S_{0}}\right)+\dots+\pred{Cut}_{G_{\epsilon}}\left(S_{k-1},\overline{S_{k-1}}\right)\right]\\
&\in&\left(1\pm\epsilon\right)\cdot\frac{1}{2}\cdot\left[\pred{Cut}_{G}\left(S_{0},\overline{S_{0}}\right)+\dots+\pred{Cut}_{G}\left(S_{k-1},\overline{S_{k-1}}\right)\right]\\
&=&\left(1\pm\epsilon\right)\cdot\pred{k-Cut}_{G}\left(S_{0},\dots,S_{k-1}\right)~.
\end{eqnarray*}

In contrast, linear-equation predicates are non-sparsifiable
for alphabet $[k]$ of size $k\ge 3$. 
Specifically, for $a\in [k]$, let the predicate $\pred{Sum}_a$ 
be satisfied by $x,y\in [k]$ iff $x+y=a \pmod k$.
Then for every positively weighted digraph $G=(V,E,w)$,
and every $\epsilon\in(0,1)$, $a\in[k]$,
every $\pred{Sum}_a$-$\epsilon$-sparsifier $G_\epsilon=(V,E_\epsilon,w_\epsilon)$ 
of $G$ must have $E=E_\epsilon$.
The argument is similar to the proof of~\theoremref{thm:And}. 
Assume for contradiction there exist $e\in E\setminus E_\epsilon$. 
Choose $x,y,z\in[k]$ that satisfy $x+y=a$, however 
the three sums $z+x$, $z+y$, $z+z$ are all not equal to $a$ (modulo $k$); 
this is clearly possible for $k\ge 4$, 
and easily verified by case analysis for $k=3$.
Consider an assignment where the endpoints of $e$ have values $x$ and $y$, respectively, and all other vertices have value $z$. 
Under this assignment, the value of $G$ is $w(e)>0$, 
while the value of $G_\epsilon$ is zero, a contradiction.

\subsection{Sketching}
\label{sec:sketching}
In \theoremref{thm:or_sketch} we showed that for every predicate $\pred{P} \in \aset{ \pred{Cut},\pred{unCut},\pred{Or}, \pred{nAnd},\overline{10}}$, 
the sketch-size of every $\epsilon$-\pred{P}-sketching-scheme  is $\Omega(n/\epsilon^2)$.

Let us now address predicates with a single $1$ in their truth table.
In the spirit of the proof of \theoremref{thm:And}, given encoding $\sk_G$ by an $\epsilon$-\pred{And}-sketching-scheme we can completely restore the graph $G$. As there are $2^{n \choose 2}$ different graphs, the sketch-size of every $\epsilon$-\pred{And}-sketching-scheme is at least $\Omega(n^2)$ bits. 
Imitating the proof of \theoremref{thm:reduction_to_And}, we can extend this lower bound to \pred{Dicut}, $01$ and $10$.

\paragraph{For-each sketches.}	
In order to reduce storage space of a sketch, one might weaken the requirements even further and allow the sketch to give a good approximation only with high probability. 
A \emph{for-each sketching scheme} is a pair of algorithms $(\sk,\est)$;
algorithm $\sk$ is a randomized algorithm that given a graph $G$ 
returns a string $\sk_G$, whose distribution we denote by $\mathcal{D}_{G}$;
algorithm $\est$ is given such a string $\sk_G$ and a subset $S\subseteq V$, 
and returns (deterministically) a value $\est(\sk_G,S)$.
We say that it is an \emph{$(\epsilon,\delta)$-$\pred{P}$-sketching-scheme} if 
\[
\forall G=(V,E,w), \forall S\subseteq V, \quad
\Pr_{\sk_G\in \mathcal{D}_{G} }\left[\est (\sk_G,S)\in\left(1\pm\epsilon\right)\cdot\pred{P}_{G}\left(S\right)\right]\ge1-\delta~.
\]

\cite{ACKQWZ15} showed that if we consider $n$-vertex graphs with weights  only in the range $[1,W]$, 
then there is an $(\epsilon,{1}/{\poly(n)})$-$\pred{Cut}$-sketching-scheme 
with sketch-size $\tilde{O}\left(n\epsilon^{-1}\cdot\log\log W\right)$ bits.
Imitating \theoremref{thm:cut_to_uncut}, we can construct  $(\epsilon,{1}/{\poly(n)})$-$\pred{P}$-sketching-scheme with the same sketch-size for every predicate \pred{P} whose truth table does not have
a single $1$ (and weights restricted to the range $[1,W]$).
A nearly-matching lower bound by \cite{ACKQWZ15} shows that 
for every $\epsilon\in(2/n,1/2)$, 
every $(\epsilon,1/10)$-$\pred{Cut}$-sketching-scheme must have sketch-size 
$\Omega(n/\epsilon)$. 
Using $\gamma$, this lower bound can be extended to \pred{unCut}. 
This technique does not work for predicates with three $1$'s in their truth table.
Fortunately, we can duplicate the proof of \cite{ACKQWZ15} while replacing \pred{Cut} by \pred{Or} and using the fact that for every two vertices $v,u$ in the graph $G$, it holds that
$\pred{Or}(\left\{ v\right\} ) 
+ \pred{Or}(\left\{ u\right\} ) 
- \pred{Or}(\left\{ v,u\right\} )
= \indic_{\aset{\left\{ u,v\right\} \in E}}$. 
We omit the details of this straightforward argument. 
A reduction from \pred{OR} using $\gamma$ and equation \ref{eq:or_by_nand} will extend the lower bound also to \pred{nAnd},$\overline{10}$ and $\overline{01}$.

Given a sketch $\sk_G$ (i.e., one sample from distribution $\mathcal{D}_G$) which encodes an $(\epsilon,\delta)$-\pred{And}-sketching-scheme, 
one can reconstruct every edge of $G$ (every bit of the adjacency matrix) 
with constant probability. 
Standard information-theoretical arguments (indexing problem) imply that the sketch-size of 
every $(\epsilon,\delta)$-\pred{And}-sketching-scheme is $\Omega(n^2)$ bits. 
Using $\gamma$ we can extend this lower bound to \pred{Dicut}, $01$ and $10$.

\subsection{\pred{unCut} Spectral Sparsifiers}
\label{sec:spectral}

Given an undirected $n$-vertex graph $G=(V,E,w)$, the Laplacian matrix is defined as  $L_G=D_G-A_G$ where $A_G$ is the adjacency matrix (i.e. $A_{i,j}=w_{i,j}=w(\{v_i,v_j\})$) and $D_G$ is a diagonal matrix of degrees (i.e. $D_{i,i}=\sum_{j\ne i}w_{i,j}$ and for $i\ne j$, $D_{i,j}=0$).
		For every $x\in\mathbb{R}^n$ it holds that $x^{t}L_{G}x=\sum_{\{v_{i},v_{j}\}\in E}w_{i,j}\cdot\left(x_{i}-x_{j}\right)^{2}$. In particular, for $\mathbf{1}_{S}$ the indicator vector of some subset $S\subseteq V$ it holds that $\mathbf{1}_{S}^{t}L_{G}\mathbf{1}_{S}=\pred{Cut}_{G}(S)$.
A subgraph $H$ of $G$ is called an $\epsilon$-\emph{spectral}-\emph{sparsifier} of $G$ if
\[
  \forall x\in \mathbb{R}^n, \quad
  x^{t}L_{H}x\in (1\pm\epsilon)\cdot x^{t}L_{G}x~.
\]
		Note that an $\epsilon$-spectral-sparsifier is in particular an $\epsilon$-\pred{Cut}-sparsifier. Nonetheless, spectral sparsifiers preserve additional properties such as the eigenvalues of the Laplacian matrix (approximately).
		\cite{BSS14} showed that every graph admits an $\epsilon$-spectral-sparsifier with $O(n/\epsilon^2)$ edges.
		
\begin{definition}
Given a graph $G$, we call $U_G= (D_G+A_G)$ the \emph{Negated Laplacian} of $G$. 
Given a subset $S\subseteq V$, let $\phi_S\in\mathbb{R}^n$ be a vector such that $\phi_{S,i}=1$ if $v_i\in S$ and $\phi_{S,i}=-1$ otherwise.
\end{definition}

		One can verify that for arbitrary $x\in\mathbb{R}^n$,
		\[
			x^{t}U_{G}x=\sum_{i<j}w_{i,j}\cdot\left(x_{i}+x_{j}\right)^{2}
		\]
		
		In particular, for every subset $S\subseteq V$,	it holds that 
		
		\[
		\phi_{S}^{t}U_{G}\phi_{S}=4\cdot \pred{unCut}_{G}(S)~.
		\]
		
		Next, we will show how we can use $U_G$ to construct an \pred{unCut}-sparsifier $G_\epsilon$ (in alternative way to \theoremref{thm:cut_to_uncut}) such that $U_{G_\epsilon}$ has (approximately) the same eigenvalues as $U_G$.
		A matrix $M\in\mathbb{R}^{n\times n}$ is called \emph{BSDD} (\emph{Balanced Symmetric Diagonally Dominant}) if $M=M^t$ and for every index $i$, $M_{i,i}=\sum_{j\ne i}|M_{i,j}|$.
		Note that $L_G$ and $U_G$ are both BSDD.
		A matrix $M'$ is \emph{governed} by $M$ if 
		whenever $M'_{i,j}\ne 0$, also $M_{i,j}\ne 0$ and has the same sign.
		Note that if $H$ is a subgraph of $G$ then $U_H$ is governed by $U_G$.
		A matrix $M'$ is called an \emph{$\epsilon$-spectral-sparsifier} of $M$ if $M'$ is governed by $M$ and
\[
  \forall x\in \mathbb{R}^n, \quad
  x^{t}M'x\in(1\pm\epsilon)\cdot x^{t}Mx~.
\]
		
		The following was implicitly shown in \cite{ACKQWZ15}.
		\begin{theorem}[\cite{ACKQWZ15}]\label{thm: SDD sparsification}
			Given BSDD matrix $M\in \mathbb{R}^{n\times n}$ and parameter $\epsilon\in(0,1)$, there is an $\epsilon$-spectral-sparsifier $M'$ for $M$ where $M'$ is BSDD matrix with $O(n/\epsilon^2)$ non-zero entries. 
		\end{theorem}

		Fix a graph $G$ and parameter $\epsilon$, according to \theoremref{thm: SDD sparsification}, there is a BSDD balanced matrix $H$ with $O(n/\epsilon^2)$ non-zero entries, that governed by $U_G$ which is a $\epsilon$-spectral-sparsifier for $U_G$.
		All this properties define a unique graph $G_\epsilon$ such that $U_{G_\epsilon}=H$. In particular $G_\epsilon$ is  $\epsilon$-\pred{unCut}-sparsifier of $G$ with  $O(n/\epsilon^2)$ edges.	
	
	\bibliographystyle{alphaurlinit}
	\bibliography{robi,csp_spars}
	
\end{document}